\documentclass[final]{dmtcs-episciences}

\usepackage[utf8]{inputenc}
\usepackage{subfigure}

\newcommand{\defparaproblem}[4]{
 \vspace{2mm}
\noindent\fbox{
 \begin{minipage}{0.96\textwidth}
 \begin{tabular*}{\textwidth}{@{\extracolsep{\fill}}lr} #1 & \\ \end{tabular*}
 {\textbf{Input:}} #2 \\
 {\textbf{Parameter:}} #3 \\
 {\textbf{Question:}} #4
 \end{minipage}
 }
 \vspace{2mm}
}

\usepackage{amstext,amsfonts,amsthm,amsmath,amssymb}
\usepackage{tikz}

\newtheorem{lemma}{Lemma}[section]
\newtheorem{theorem}[lemma]{Theorem}
\newtheorem{claim}[lemma]{Claim}
\newtheorem{corollary}[lemma]{Corollary}
\newtheorem{observation}[lemma]{Observation}

\DeclareMathOperator{\tpw}{\textup{\textbf{tpw}}}
\DeclareMathOperator{\tcw}{\textup{\textbf{tcw}}}
\DeclareMathOperator{\tw}{\textup{\textbf{tw}}}

\DeclareMathOperator{\wn}{\textup{\textbf{wn}}}
\DeclareMathOperator{\gn}{\textup{\textbf{gn}}}
\DeclareMathOperator{\cut}{cut}

\author[Hans L. Bodlaender \and Carla Groenland \and Hugo Jacob]{Hans L. Bodlaender\affiliationmark{1}
  \and Carla Groenland\affiliationmark{2}
  \and Hugo Jacob\affiliationmark{3}}
\title[On the parameterized complexity of computing tree-partitions]{On the parameterized complexity of computing tree-partitions}
\affiliation{
  Universiteit Utrecht, Utrecht, Netherlands\\
  Technische Universiteit Delft, Delft, Netherlands\\
  LIRMM, Université de Montpellier, CNRS, Montpellier, France}
\keywords{parameterized algorithms, tree-partitions, tree-partition-width, tree-cut width, domino tree\-width, tree\-width, approximation algorithms, parameterized complexity}

\begin{document}
\publicationdata{vol. 26:3}{2024}{3}{10.46298/dmtcs.12540}{2023-11-12; 2023-11-12; 2024-04-30; 2024-10-03}{2024-07-29}

\maketitle

\begin{abstract}
    We study the parameterized complexity of computing the tree-partition-width, a graph parameter equivalent to tree\-width on graphs of bounded maximum degree. 
    
    On one hand, we can obtain approximations of the 
    tree-partition-width efficiently: we show that there is an algorithm that, given an $n$-vertex graph $G$ and an integer $k$, constructs a tree-partition of width $O(k^7)$ for $G$ or reports that $G$ has tree-partition-width more than $k$, in time $k^{O(1)}n^2$. We can improve slightly on the approximation factor by sacrificing the dependence on $k$, or on $n$.

On the other hand, we show the problem of computing
tree-partition-width exactly is XALP-complete, which implies that it is $W[t]$-hard for all $t$. We deduce XALP-completeness of the problem of computing the domino treewidth. 

Next, we adapt some known results on the parameter tree-partition-width and the topological minor relation, and use them to compare tree-partition-width to tree-cut width.

Finally, for the related parameter weighted
tree-partition-width, we give a similar approximation algorithm (with ratio now $O(k^{15})$)
and show XALP-completeness for the special case where vertices
and edges have weight 1.
\end{abstract}

\section{Introduction}

Graph decompositions have been a very useful tool to draw the line between tractability and intractability of computational problems. There are many meta-theorems showing that a collection of problems can be solved efficiently if a decomposition of some form is given (e.g. for treewidth \cite{Courcelle90}, for clique-width \cite{CourcelleMR00}, for twin-width \cite{tww1}, for mim-width \cite{mimwidthmetathm}). By finding efficient algorithms to compute a decomposition if it exists, we deduce the existence of efficient algorithms even if the decomposition is not given.
In particular, this proves useful when designing win-win arguments: for some problems, the existence of a solution and the existence of a decomposition are not independent, so that we can either use the decomposition for an efficient computation of the solution, or conclude that a solution must (or cannot) exist when there is no decomposition of small enough width.

The most successful notion of graph decomposition to date is certainly tree decompositions, and its corresponding parameter treewidth. 
Any problem expressible in MSO$_2$\footnote{Formulae with quantification over sets of edges or vertices, quantification over vertices and edges, and with the incidence predicate.} can be solved in linear time in graphs of bounded treewidth due to a meta-theorem of Courcelle \cite{Courcelle90} and the algorithm of Bodlaender for computing an optimal tree decomposition \cite{Bodlaender96}.
Treewidth is a central tool in the study of minor-closed graph classes. A minor-closed graph class has bounded treewidth if and only if it contains no large grid minor.

In this paper, we focus on the parameter tree-partition-width (also called strong treewidth) which was independently introduced by Seese \cite{Seese85} and Halin \cite{Halin91}. It is known to have simple relations to treewidth \cite{DingO95,Wood09}: $\tw = O(\tpw)$, and $\tpw = O(\Delta \tw)$, where $\tw,\tpw,\Delta$ denote the treewidth, the tree-partition-width, and the maximum degree respectively. 
Applications of tree-partition-width include graph drawing and graph coloring \cite{CarmiDMW08,GiacomoLM05,DujmovicMW05,DujmovicSW07,WoodT06,BaratW08,AlonDV03}.
Recently, Bodlaender, Cornelissen and Van der Wegen~\cite{BodlaenderCW22} showed for a number of problems (in particular, problems related to network flow) that these are intractable (XNLP-complete) when the
pathwidth is used as parameter, but become fixed parameter tractable when parameterized by the width of a given tree-partition. This raises the question of the complexity of finding tree-partitions. We show that computing tree-partitions of approximate width is tractable. 
\begin{theorem}
\label{thm:approx}
There is an algorithm that given an $n$-vertex graph $G$ and an integer $k$, constructs a tree-partition of width $O(k^7)$ for $G$ or reports that $G$ has tree-partition-width more than $k$, in time $k^{O(1)}n^2$.
\end{theorem}
Thus, this removes the requirement from the results from \cite{BodlaenderCW22} that a tree-partition
of small width is part of the input. Our technique is modular and allows us to also give alternatives running in FPT time or polynomial time with an improved approximation factor (see Theorem \ref{thm:main}). Although not formulated as an algorithm, a construction of Ding and Oporowski \cite{DingO96} implies an FPT algorithm to compute tree-partitions of width $f(k)$ for graphs of tree-partition-width $k$, for some fixed computable function $f$. We adapt their construction and give some new arguments designed for our purposes. This significantly improves on the upper bounds to the width, and the running time. 

\smallskip

The results from \cite{BodlaenderCW22} are stated in terms of the notions of stable gonality,
stable tree-partition-width and a new parameter
called weighted tree-partition-width\footnote{In earlier versions of \cite{BodlaenderCW22}, the parameter was called treebreadth, but to avoid confusion, the term weighted tree-partition-width is now used.}.
The notion of stability comes from the algebraic geometry origins of the notion of gonality;
in graph-theoretic terms, this implies that we look at the minimum over all possible subdivisions
of edges. It turns out
that tree-partition-width, stable tree-partition-width, 
and weighted tree-partition-width (with edge weights one) are bounded by polynomial functions of each other; see Section~\ref{section:tcw} and Corollary~\ref{coro:stable-tpw}. 
In Section~\ref{section:weightedtpw}, we obtain
some results on the complexity of computing and approximating the weighted tree-partition-width, as corollaries of earlier results.

\smallskip

Related to tree-partition-width is the notion of domino treewidth, first studied by Bodlaender and Engelfriet~\cite{BodlaenderE97}. A \textit{domino tree decomposition} is a tree decomposition where each vertex is in at most
two bags. 
Where graphs of small tree-partition-width can have large degree, a graph of domino treewidth $k$ has maximum degree at 
most $2k$. Bodlaender and Engelfriet show that \textsc{Domino Treewidth} is hard for each class $W[t]$, $t\in \mathbb{N}$; we improve this result and show XALP-completeness.
\begin{theorem}
\label{thm:xalpcomplthms}
\textsc{Domino Treewidth} and 
\textsc{Tree-Partition-Width} are XALP-complete.
\end{theorem}
In \cite{Bodlaender99}, Bodlaender gave an algorithm to compute a domino tree decomposition of width $O(\tw\Delta^2)$ in $f(\tw)n^2$ time for $n$-vertex graphs of treewidth $\tw$ and maximum degree $\Delta$, where $f$ is a fixed computable function. This implies an approximation algorithm for domino treewidth.

We also consider the parameter \emph{tree-cut width} introduced by Wollan in \cite{Wollan15}.
As the tractability results of Bodlaender et al.~\cite{BodlaenderCW22} use techniques similar to a previous work on algorithmic applications of tree-cut width \cite{Ganian0S15}, one may wonder whether there is a relationship between tree-cut width and tree-partition-width.

We show the following results.
\begin{itemize}
    \item We obtain a parameter that is polynomially tied to tree-partition-width and is topological minor monotone (see Theorem~\ref{thm:tpw-obstructions}).  We use this to show that tree-partition-width is relatively stable with respect to subdivisions: if we define $\underline{\tpw}(G)$ (resp. $\overline{\tpw}(G)$) as the minimum (resp. maximum) tree-partition-width  over subdivisions of $G$, then $\underline{\tpw},~\tpw$ and $\overline{\tpw}$ are polynomially tied. The parameter $\underline{\tpw}(G)$ corresponds to `stable tree-partition-width'.
    \item We show tree-partition-width is polynomially upper bounded by tree-cut width (see Theorem~\ref{thm:tpw-tcw}) by relating the tree-cut width to the tree-partition-width of a subdivision. 
    \item On the other hand, a bound on tree-partition-width does not imply a bound on tree-cut width (see Observation~\ref{obs:tcwK3n}).
\end{itemize}

\paragraph*{Paper overview}
In Section~\ref{section:approx}, we provide our results on approximating the tree-partition-width.
In Section~\ref{section:exact}, we show that computing the tree-partition-width is XALP-complete. We then derive XALP-completeness of computing the domino treewidth. In Section~\ref{section:tcw}, we give our results relating 
tree-cut width to tree-partition-width.
In Section~\ref{section:weightedtpw}, we give the
results for weighted tree-partition-width.
Some concluding remarks are made in Section~\ref{section:conclusion}.

\section{Preliminaries}
\label{section:preliminaries}
The set of positive integers is denoted by
$\mathbb{Z}^+$; the set of non-negative integers
is denoted by $\mathbb{N}$.

A \emph{tree-partition} of a graph $G=(V,E)$ is a tuple $(T, (B_i)_{i\in V(T)})$, where $B_i \subseteq V(G)$, with the following properties.
\begin{itemize}
    \item $T$ is a tree.
    \item For each $v\in V$ there is a unique $i(v)\in V(T)$ such that $v\in B_{i(v)}$.
    \item For any edge $uv\in E$, either $i(v)=i(u)$ or $i(u)i(v)$ is an edge of $T$. 
\end{itemize}
The \emph{size} of a \emph{bag} $B_i$ is $|B_i|$, the number of vertices it contains.
The \emph{width} of the decomposition is given by $\max_{i \in V(T)} |B_i|$. The tree-partition-width ($\tpw$) of a graph $G$ is the minimum width of a tree-partition of $G$.

We also consider a variant of the notion for weighted graphs. The notion was introduced
in~\cite{BodlaenderCW22}. We use a slight generalization where we allow also weights of vertices, to facilitate some of our proofs.

Let $G=(V,E)$ be a graph, with a weight function
for vertices $w_V: V \rightarrow \mathbb{Z}^+$,
and a weight function $w_E: E \rightarrow \mathbb{N}$. 
The \emph{breadth} of a tree-partition
$(T, (B_i)_{i\in V(T)})$
of $G$ is 
the maximum over total weights of bags: $$\max_{i\in V(T)} \sum_{v\in B_i} w_V(v)$$ and total weights of edge cuts between pairs of adjacent bags: $$\max_{ii'\in E(T)} \sum_{vw\in E, v\in B_i, w'\in B_{i'}} w_E(vw)$$
The \emph{weighted tree-partition-width} of a
graph $G$ 
with vertex and edge weights is the minimum
breadth over all tree-partitions of $G$.

In other words, we take the maximum over all bags of the total weight of the vertices in the bag,
and maximum over all pairs of adjacent bags of the total weight of edges between these bags; and then, we take the maximum over these two values. The
tree-partition-width of a graph equals the weighted
tree-partition-width where all vertices have weight 1 and all edges have weight 0.

A {\em tree decomposition} of a graph $G=(V,E)$ is a pair $(T=(I,F)$, $\{X_i~|~ i\in T\})$ with
$T=(I,F)$ a tree and $\{X_i~|~i\in I\})$ a family of (not necessarily disjoint)
subsets of $V$ (called {\em bags}) such that $\bigcup_{i\in I} X_i = V$,
for all edges $vw\in E$, there is an $i$ with $v,w\in X_i$, and for all
$v$, the nodes $\{i\in I~|~v\in X_i \}$ form a connected subtree of $T$.
The {\em width} of a tree decomposition $(T, \{X_i~|~ i\in T\})$ is
$\max_{i\in I} |X_i|-1$, and the {\em treewidth} ($\tw$) of a graph $G$ 
is the minimum width over all tree decompositions of $G$. The \emph{domino treewidth} is the minimum width over all tree decompositions of $G$ such that each vertex appears in at most two bags.

We say that two parameters $\alpha,\beta$ are \emph{(polynomially) tied} if there exist (polynomial) functions $f,g$ such that $\alpha \leq f(\beta)$ and $\beta \leq g(\alpha)$.

\section{Approximation algorithm for tree-partition-width}\label{section:approx}

We first describe our algorithm, then prove correctness and finally discuss the trade-offs between running time and solution quality.

\subsection{Description of the algorithm}\label{subsec:scheme}
Let $G$ be a graph, and $k$ any positive integer.
We describe a scheme that produces a tree-partition of $G$ of width $O(wbk^3)=$
$k^{O(1)}$, or reports that $\tpw(G) > k$. We will use various different functions of $k$ for $b$ and $w$, depending on the quality/time trade-offs of the black-box algorithms inserted into our algorithm (e.g. for approximating treewidth).

\subparagraph*{Step 1} \emph{Compute a tree decomposition for $G$ of width $w(k)$ or conclude that $\tpw(G)>k$.}\\
As mentioned above, we do not directly specify the function $w=w(k)$, since different algorithms for step 1 give different solution qualities (bounds for $w(k)$) and running times. Since $\tw+1\leq 2\cdot \tpw$, if $\tw(G)>2k-1$ it follows that $\tpw(G)>k$.
Suppose that we obtained a decomposition of width $w$ (which could be larger than $\tw(G)$). If $\tpw(G)\leq k$, then $\tw(G) \leq 2k-1$, and so there are at most $(2k-1)n$ edges in $G$. If $G$ has more than $(2k-1)n$ edges, we directly reject.

We set a threshold $b\geq \max\{2k-1,w+1\}$. We define an auxiliary graph $G^b$ as follows. The vertex set of $G^b$ is $V(G)$.
The edges of $G^b$ are given by the pairs of vertices $u,v\in V(G)$ with minimum $u$-$v$ separator of size at least $b$.

\subparagraph*{Step 2} \textit{Construct the auxiliary graph $G^b$ with connected components of size at most $k$ or report that $\tpw(G) > k$.}\\
We later describe several ways of computing the edges of $G^b$.
We will show in Claim~\ref{claim:large-cuts} that vertices in the same connected component of $G^b$ must be in the same bag for any tree-partition of width at most $k$. For this reason, we conclude that $\tpw(G) > k$ if a  component of $G^b$ has more than $k$ vertices.

We define $H$, the \emph{$b$-reduction} of $G$, which is the graph obtained from $G$ by identifying the connected components of $G^b$.

\subparagraph*{Step 3} \textit{We compute a tree decomposition of width $w$ for each 2-connected component of $H$.}\\
Given the components of $G^b$, we can compute $H$, and its 2-connected components in time $O(k^{O(1)}n)$. Using Claim~\ref{claim:dec-red}, we obtain a tree decomposition of $H$ by replacing vertices of $G$, in the tree decomposition of $G$, by their component in $G^b$.

By Claim~\ref{claim:degree-bound}, the maximum degree $\Delta_H$ within the 2-connected components of $H$ is at most $Cbk^2$ for some constant $C$ when $\tpw(G) \leq k$.

\subparagraph*{Step 4} \textit{If $\Delta_H>Cbk^2$, report $\tpw(G)>k$. Else, compute a tree partition of width $O(w\Delta_H)=O(wbk^2)$ for $H$.}\\
By rooting the decomposition of $H$ in 2-connected components, we can define a \emph{parent} cut-vertex for each 2-connected component except the root.
We separately compute tree partitions for each 2-connected component of $H$ with the constraint that their parent cut-vertex should be the single vertex of its bag. A construction of Wood \cite{Wood09} enables us to compute a tree-partition of width $O(\Delta w)$ for any graph of maximum degree $\Delta$ and treewidth $w$; this can be adjusted to allow for this isolation constraint without increasing the upper bound on the width. We give the details of this in Corollary~\ref{corol:combine-blocks}. After doing this, the partitions of each component can be combined without increasing the width. Indeed, although cut-vertices are shared, only one 2-connected component will consider putting other vertices in its bag. We obtain a tree-partition of $H$ of width $O(wbk^2)$.

\subparagraph*{Step 5} \textit{Deduce a tree partition of width $O(wbk^3)$ for $G$.}\\
We `expand' the vertices of $H$. In the tree-partition of $H$, each vertex of $H$ is replaced by the vertices of the corresponding connected component of $G^b$. This gives a tree-partition of $G$ of width $O(wbk^3)$.

\subsection{Correctness}\label{subsec:scheme-correctness}
For $s,t\in V(G)$ we denote by $\mu(s,t)$ the size of a minimum $s$-$t$ separator in $G-st$.
\begin{claim}\label{claim:large-cuts}
Let $G$ be a graph and $s,t\in V(G)$.
    \begin{itemize}
        \item If $\mu(s,t) \geq k+1$, then in any tree-partition of width at most $k$, $s$ and $t$ must be in adjacent bags or the same bag;
        \item If $\mu(s,t) \geq 2k-1$, then in any tree-partition of width at most $k$, $s$ and $t$ must be in the same bag.
    \end{itemize}
\end{claim}

\begin{proof}
    Assume that $s$ and $t$ are not in adjacent bags nor in the same bag of a tree-partition of width at most $k$, then any internal bag on the path between their respective bags is an $s$-$t$ separator. In particular, $\mu(s,t) \leq k$.
    This proves the first point by contraposition.

    Assume that $s$ and $t$ are in adjacent bags but not in the same bag for some tree-partition of width at most $k$. We denote their respective bags by $B_s$ and $B_t$. Then, $(B_s \cup B_t) \setminus \{s,t\}$ is an $s$-$t$ separator of $G - st$. Consequently, $\mu(s,t) \leq 2k-2$.
    This proves the second point.
\end{proof}

\begin{claim}\label{claim:dec-red}
    Consider $(T,(X_i)_{i \in V(T)})$ a tree decomposition of width $w$ of $G$, $b \geq w+1$, and let $Y_i$ be the set of connected components of $G^b$ that intersect with $X_i$. Then $(T,(Y_i)_{i \in V(T)})$ is a tree decomposition of the $b$-reduction $H$ of $G$.
\end{claim}

\begin{proof}
    Every component of $G^b$ appears in at least one $Y_i$, because it contains a vertex which must appear in at least one $X_i$. Furthermore, for each edge $UV$ of $H$, there must be vertices $u \in U, v \in V$ such that $uv$ is an edge of $G$. Hence, there is a bag $X_i$ containing $u$ and $v$ so $Y_i$ contains $U$ and $V$. Finally, suppose that there is a bag $Y_i$ not containing a component $C$ of $G^b$, and several components of $T-i$ have bags containing $C$. There must be an edge of $G^b$ connecting vertices $u$ and $v$ of $C$ such that $u$ is in bags of $X$ and $v$ is in bags of $X'$, where $X$ and $X'$ are in different components of $T-i$. By definition of $G^b$, the minimal size of a separator of $u$ and $v$ in $G$ is at least $b\geq w+1$. However, since the tree decomposition $(T,(X_i))$ has width $w$ and the bags containing $u$ are disjoint from the bags containing $v$ (in particular $X_i$ separates them), there is a separator of $u$ and $v$ of size at most $w$, a contradiction. This concludes the proof that $(T,(Y_i))$ is a tree decomposition of $H$.
\end{proof}

\begin{claim}\label{claim:degree-bound}
    If $H$ is the $b$-reduction of $G$, $\tpw(G) \leq k$, and $B$ is one of its 2-connected components, then the maximum degree in $B$ is $O(bk^2)$.
\end{claim}

\begin{proof}
    Consider $u$ a vertex achieving maximum degree in $B$. By definition of $B$, $B-u$ is connected.
    We denote by $N$ the neighborhood of $u$ in $B$.
    Let $T$ be a spanning tree of $B-u$. We iteratively remove leaves that are not in $N$. This produces the reduced tree $T'$.
    The maximum degree in this tree is $b-1$ as the set of edges incident to a given vertex can be extended to disjoint paths leading to vertices in $N$, hence leading to $u$. Since $H$ is a $b$-reduction, the number of disjoint paths between two vertices must be less than $b$.

    Clearly the neighbors of $u$ must be either in the same bag as $u$ or in a neighboring bag. Since the bag of $u$ will be a separator of vertices that are in distinct neighboring bags, in particular, it splits the graph into several components each containing at most $k$ neighbors of $u$.

    There must exist a subset of vertices of $T'$ of size at most $k-1$ whose removal splits $T'$ in components containing vertices of $N$ of total weight at most $k$. Since the degree of a vertex of $T'$ is at most $b-1$, removing one of its vertices adds at most $b-2$ new components. Hence, after removing $k - 1$ vertices, there are at most $1+(k-1)(b-2)$ components. We conclude that $|N| \leq k(1 + (k-1)(b-2))$. Since $u$ had maximum degree in $B$, we conclude.
\end{proof}

In \cite{Wood09}, Wood shows the following lemma.

\begin{lemma}
\label{lem:wood}
    Let $\alpha = 1 + 1/\sqrt{2}$ and $\gamma = 1 +\sqrt{2}$. Let $G$ be a graph with treewidth at most $k \geq 1$ and maximum degree at most $\Delta \geq 1$. Then $G$ has tree-partition-width $\tpw(G) \leq \gamma (k+1) (3\gamma\Delta -1)$.

    Moreover, for each set $S \subseteq V(G)$ such that $(\gamma+1)(k+1) \leq |S| \leq 3(\gamma + 1)(k+1)\Delta$, there is a tree-partition of $G$ with width at most $\gamma (k+1)(3\gamma\Delta -1)$ such that $S$ is contained in a single bag containing at most $\alpha|S| - \gamma(k+1)$ vertices.
\end{lemma}

We deduce this slightly stronger version of \cite[Theorem 1]{Wood09}

\begin{corollary}\label{corol:combine-blocks}
    From a tree decomposition of width $w$ in a graph $G$ of maximum degree $\Delta$, for any vertex $v$ of $G$, we can produce a tree-partition of $G$ of width $O(\Delta w)$ in which $v$ is the only vertex of its bag.
\end{corollary}

\begin{proof}
We wish to apply Lemma \ref{lem:wood} to $S \supseteq N(v)$. Let $\gamma=1+\sqrt{2}$.
We have $|N(v)| \leq \Delta$, so in particular, $|N(v)| \leq 3(\gamma +1)(w+1)\Delta$. In case $|N(v)| < (\gamma +1)(w+1)$, we can add arbitrary vertices to $N(v)$ to form $S$ satisfying $|S| \geq (\gamma +1)(w+1)$. Otherwise, we simply set $S=N(v)$.
    We then apply the lemma to $S$ in $G-v$.
    There is a single bag that contains $N(v)$, and so we may add the bag $\{v\}$ adjacent to this in order to deduce a tree partition of $G$ of width $O(\Delta w)$ in which $v$ is the only vertex of its bag.
\end{proof}

\subsection{Time/quality trade-offs}

For Step 1, we consider the following algorithms to compute tree decompositions:
\begin{itemize}
    \item An algorithm of Korhonen \cite{Korhonen21} computes a tree decomposition of width at most $2k+1$ or reports that $\tw(G)>k$ in time $2^{O(k)}n$.
    \item An algorithm of Fomin et al. \cite{FominLSPW18} computes a tree decomposition of width $O(k^2)$ or reports that $\tw(G)>k$ in time $O(k^7n\log n)$.
    \item An algorithm of Feige et al. \cite{FeigeHL08} computes a tree decomposition of width $O(k \sqrt{\log k})$ or reports that $\tw(G)>k$ in time $O(n^{O(1)})$.
\end{itemize}

Recall that we denote by $w$ the width of the computed tree decomposition of $G$.

We give two methods to compute $G^b$ in step 2 of the algorithm.
\begin{itemize}
    \item We can use a maximum-flow algorithm (e.g. Ford-Fulkerson \cite{FordFulkerson}) to compute for each pair $\{s,t\}$ of vertices of $G$ whether there are at least $b$ vertex disjoint paths from $s$ to $t$, in time $O(bkn)$. To compute a minimum vertex cut, replace each vertex $v$ by two vertices $v_{\text{in}},v_{\text{out}}$ with an arc from $v_{\text{in}}$ to $v_{\text{out}}$. All arcs going to $v$ should go to $v_{\text{in}}$, and all arcs leaving $v$ should leave $v_{\text{out}}$. All arcs are given capacity 1. We may stop the maximum flow algorithm as soon as a flow of at least $b$ was found. Furthermore, we can reduce the number of pairs $\{s,t\}$ of vertices to check to $O(w n)$, as each pair must be contained in a bag due to $b \geq w+1$. This results in a total time of $O(wbkn^2)$.

    \item We can also use dynamic programming to enumerate all possible ways of connecting pairs of vertices that are in the same bag in time $2^{O(w \log w)}n$, which is sufficient to compute $G^{b}$. A state of the dynamic programming consists of the subset of vertices of the bag that are used by the partial solution, a matching on some of these vertices, up to two vertices that were decided as endpoints of the constructed paths, the number of already constructed paths between the endpoints, and two disjoint subsets of the used vertices that are not endpoints, nor in the matching such that we found a disjoint path from the first or second endpoint to them. The bound of $2^{O(w \log w)}$ follows from the fact that this is a subset of the labeled forests on $w$ vertices. We may assume that our tree decomposition is rooted and binary. We first tabulate answers for each subtree of the decomposition by starting from the leaves, and then tabulate answers for each complement of a subtree by starting from the root and, when branching to some child, combining with the partial solutions of the subtree of the other child. By combining tabulated values for subtrees and their complements, we obtain the sought information.
\end{itemize}

The $b$-reduction $H$ of $G$ and its 2-connected components can be computed in $O(k^{O(1)}n)$ time (see e.g. \cite{Biconnected}), since the size of the graph is $O(k^{O(1)}n)$ here.

We will now make use of the following result due to Bodlaender and Hagerup \cite{BodlaenderH98}:

\begin{lemma}\label{lem:balanced-tw}
There is an algorithm that given a tree decomposition of width $k$ with $O(n)$ nodes of a graph $G$, finds a rooted binary tree decomposition of $G$ of width at most $3k+2$ with depth $O(\log n)$ in $O(kn)$ time.
\end{lemma}

When implementing the construction of Wood for 2-connected components of $H$, the running time is dominated by $O(n)$ queries to find a balanced separator with respect to a set $S$ of size $k^{O(1)}$. After a preprocessing in time $O(kn)$, we can do this in time $k^{O(1)}d$ where $d$ is the diameter of our tree decomposition.
We first obtain a binary balanced decomposition using Lemma~\ref{lem:balanced-tw}, then reindex the vertices in such a way that we can check if a vertex is in some bag of a given subtree of tree decomposition in constant time. Using this, we can in time $k^{O(1)}$ determine whether a bag is a balanced separator of $S$, and if not move to the subtree containing the most vertices of $S$. This procedure will consider at most $d$ bags, hence the total running time of $k^{O(1)}d$. Since the decomposition has depth $O(\log n)$, it also has diameter $d=O(\log n)$. Hence, the construction of Wood can be executed in time $k^{O(1)}n \log n$.

\begin{lemma}
\label{lem:greedywood}
    We can compute a tree partition of width $O(\Delta w)$ in time $O(k^{O(1)}n \log n)$ when given a tree decomposition of width $w=k^{O(1)}$.
\end{lemma}

To improve on the function of $n$ in the running time of our procedure to compute tree-partitions of width $O(\Delta w)$, we can use some of the techniques introduced in \cite{BodlaenderDDFLP16}. If we use separators that are balanced with respect to the subgraph that has to be decomposed, we obtain a balanced decomposition as observed by Reed \cite{Reed92} which gave an approximation algorithm for tree decompositions with running time $f(k)n\log n$. If, in addition, we stop processing once we reach components of size $O(\log n)$, we have computed at most $f(k)n/\log n$ separators which can each be found in time $f(k)\log n$ using the data structure introduced in \cite{BodlaenderDDFLP16}. This means that getting to components of size $O(\log n)$ takes time $f(k)n$. On each of the obtained components, we can either apply our previous $f(k) n \log n$ construction, which leads to an $f(k)n \log \log n$ algorithm, or apply our construction recursively to obtain an $f(k) n \log^{(\alpha)} n$ algorithm where $\log^{(\alpha)}$ is the $\alpha$-fold composition of $\log$.

\begin{lemma}
\label{lem:logalpha}
    A balanced tree-partition of width $O(\Delta w)$ can be computed in time $2^{O(k)}n \log^{(\alpha)} n$, for all integers $\alpha \geq 1$, for a graph of maximum degree $\Delta = k^{O(1)}$ when given a tree decomposition of width $w=O(k)$.
\end{lemma}

\begin{proof}
    We first observe that to make the decomposition balanced, we only need to add a balanced separator once per bag of the tree-partition which still gives a width of $O(\Delta w)$.
    For each bag of the constructed tree-partition, we compute a balanced separator once and we compute $O(\Delta)$ $S$-balanced separators. These can be computed in time $2^{O(k)}\log n$ by the data structure
    after its initialization in time $2^{O(k)}n$. One might worry that because we look at sets $S$ of size polynomial in $k$ and not linear unlike \cite{BodlaenderDDFLP16}, the data structure does not work or has a running time $2^{O(|S|)}$ instead of $2^{O(k)}$. However, the exponential part in the analysis is an exponential in the width of the tree decomposition. The size of $S$ does have an impact on the running time, but only appears in a polynomial factor. Since we bound the size of $S$ by a polynomial in $k$, the polynomial in $|S|$ still gives a polynomial in $k$ in our setting, and it is still dominated by the exponential in $k$.

    We prove the existence of an algorithm of running time $O(2^{O(k)}n \log^{(\alpha)}n)$ for every $\alpha \geq 1$ by induction on $\alpha$.

    First, we initialize the data structure to compute separators of \cite{BodlaenderDDFLP16}.
    Then we use this data structure to compute separators. We will compute only $k^{O(1)}$ separators per bag, each in time $2^{O(k)}\log n$, which takes $2^{O(k)}\log n$ time per bag.
    If $\alpha=1$, we fully process subinstances and obtain a running time of $2^{O(k)}n \log n$.
    Otherwise, we stop processing subinstances once they have size $O(\log n)$.
    If we stop processing subinstances when they reach size $O(\log n)$, we compute only $O(n /\log n)$ bags because the tree-partition is balanced see \cite[Lemma 4.3 and Claim 4.4]{BodlaenderDDFLP16}.
    We then compute new tree decompositions for each of the $O(\log n)$ size components in total time $2^{O(k)}n$. For each of them, we apply the algorithm with $\alpha'=\alpha-1$, the total running time is then bounded by $2^{O(k)}n \log^{(\alpha')} \log n = 2^{O(k)}n \log^{(\alpha)} n$.
\end{proof}

The above algorithm can be turned into an algorithm running in time $2^{O(k)}n$ using the following observation: the number of different configurations to be solved for components of size $O(\log\log(n))$ is small enough to allow us to enumerate all distinct configurations in sublinear time.

\begin{claim}
\label{claim:enum-loglog}
    We can enumerate all configurations of pairs $(H,S)$ with $H$ a graph on $O(\log\log(n))$ vertices, $S$ a subset of vertices of $H$, compute a tree-partition of $H$ having a superset of $S$ as its root bag for each configuration, and tabulate all results in time $o(n)$ in some table $T$.
\end{claim}

\begin{proof}
    The number of pairs $(H,S)$ is $2^{O((\log\log(n))^2)}\cdot2^{O(\log\log(n))}=2^{O((\log\log(n))^2)}$.
    Let $n'$ be the number of vertices of $H$. If $k \geq n'$, we may use a single bag containing all vertices of $H$.
    Otherwise, for each such pair, we can compute a tree-partition using the algorithm running in time $2^{O(k)}n' \log n'$. Using the fact that $k \leq n'$, we can bound this running time by $2^{O(n')}n' \log n'=2^{O(n')}=2^{O(\log\log(n))}$.
    We can then store the tree-partition of $H$ which has size $O(n')$ in table entry $T[H,S]$.

    The total computation time and space is bounded by $2^{O((\log\log(n))^2)}=o(n)$
\end{proof}

\begin{lemma}\label{lem:linear}
    A balanced tree-partition of width $O(\Delta w)$ can be computed in time $2^{O(k)}n$ for a graph of maximum degree $\Delta = k^{O(1)}$ when given a tree decomposition of width $w=O(k)$.
\end{lemma}

\begin{proof}
    We first compute suitable decompositions for all pairs $(H,S)$ of graphs on $O(\log\log(n))$ vertices and subsets of vertices as described in Claim~\ref{claim:enum-loglog}.
    We then recursively decompose similarly to Lemma~\ref{lem:logalpha} with $\alpha=2$ except that now when reaching components of size $O(\log\log(n))$, we simply have to look at a relevant entry of table $T$ to obtain a decomposition for the component.

    This is done as follows. We first compute an arbitrary bijection of the vertex set of the component with $[1,n']$ where $n'$ is the size of the component. It is then straightforward to obtain an adjacency list describing the component and the list of vertices with neighbors outside the component, both using the new indices. Using this description, we obtain a pair $(H,S)$ allowing us to lookup a decomposition at entry $T[H,S]$. We then obtain a tree-partition with the new indices, and only have to replace all indices by the original indices of vertices in the component. This procedure takes time linear in the size of the component.
    
    The overall running time of the algorithm is $2^{O(k)}n$, since we spend $n+o(n)$ time in total for components of size $O(\log\log(n))$, $2^{O(k)}n\log(n)/\log(n)=2^{O(k)}n$ to reduce the graph to components of size $O(\log(n))$, and $2^{O(k)}(\sum |C_i|\log|C_i|/\log|C_i|)=2^{O(k)}n$ to reduce each component $C_i$ of size $O(\log(n))$ to components of size $O(\log\log(n))$.
\end{proof}

By combining the previous algorithms we obtain the following theorem.

\begin{theorem}
\label{thm:main}
    There is a polynomial time algorithm that constructs a tree-partition of width $O(k^5\log k)$ or reports that the tree-partition-width is more than $k$.

    There is an algorithm running in time $2^{O(k \log k)}n$ that computes a tree-partition of width $O(k^5)$ or reports that the tree-partition-width is more than $k$.

    There is an algorithm running in time $k^{O(1)}n^2$ that computes a tree-partition of width $O(k^7)$ or reports that the tree-partition-width is more than $k$.
\end{theorem}

\begin{proof}
    The first algorithm uses the algorithm of Feige et al. to compute the tree decomposition, then naively computes $G^b$, and then finds balanced separators for Wood's construction using the tree decomposition in polynomial time (no need to balance the decomposition).

    The second algorithm uses Korhonen's algorithm to compute the tree decomposition, then computes $G^b$ using the dynamic programming approach, and then applies Lemma~\ref{lem:linear} to implement the tree-partition construction on each 2-connected component of $H$.

    The third algorithm uses the algorithm of Fomin et al. to compute the tree decomposition, then computes $G^b$ via a maximum-flow algorithm in time $O(wbkn^2)=O(k^5n^2)$, and then computes the tree-partition for each 2-connected component of $H$ using Lemma~\ref{lem:greedywood}.

    The guarantees on the width follow from the analysis of our scheme (see \ref{subsec:scheme} and \ref{subsec:scheme-correctness}).
\end{proof}

\section{XALP-completeness of Tree-Partition-Width}\label{section:exact}

In this section, we show that the {\sc Tree-Partition-Width} problem is XALP-complete, even when we use the target width and the degree as combined parameter.
As a relatively simple consequence, we obtain that {\sc Domino Treewidth} is XALP-complete.

XALP is the class of all parameterized problems that can be solved in $f(k)n^{O(1)}$ time and $f(k)\log n$ space on a nondeterministic Turing Machine with access to a push-down automaton, or equivalently the class of problems that can be solved by an alternating Turing Machine in $f(k)n^{O(1)}$ treesize and $f(k)\log(n)$ space. An alternating Turing Machine (ATM) is nondeterministic Turing Machine with some extra states where we ask for all of the transitions to lead to acceptance. This creates independent configurations that must all lead to acceptance, and we call `co-nondeterministic step' the process of obtaining these independent configurations.

XALP is closed by reductions using at most $f(k)\log n$ space and running in FPT time. These two conditions are implied by using at most $f(k) + O(\log n)$ space. We call reductions respecting the latter condition \emph{parameterized log-space} reductions (or pl-reductions).

This class is relevant here because the problems we consider are complete for it. Completeness for XALP has the following consequences: W[$t$]-hardness for all positive integers $t$, membership in XP, and there is a conjecture that XP space is required for algorithms running in XP time. If the conjecture holds, this roughly means that the dynamic programming algorithm used for membership is optimal.

\begin{lemma}
{\sc Tree-Partition-Width} is in XALP.
\label{lemma:tpwinxalp}
\end{lemma}

\begin{proof}
To keep things simple, we will use as a black box the fact that reachability in undirected graphs can be decided in log-space \cite{Reingold}.
We assume that the vertices have some arbitrary ordering $\sigma$.

For now, assume that the given graph is connected.

We begin by guessing at most $k$ vertices to form an initial bag $B_0$, and have an empty parent bag $P_0$ .
We will recursively extend a partial tree-partition in the following manner.
Suppose that we have a bag $B$ with parent bag $P$, we must find a child bag for $B$ in each connected component of $G-B$ that does not contain a vertex of $P$.
We use the fact that a connected component can be identified by its vertex appearing first in $\sigma$, that the restriction of $\sigma$ to these representatives gives an ordering on $\sigma$, and that we can compute such representatives in log-space.
Let us denote by $c$ the current vertex representative of a connected component of $G-B$.
$c$ is initially the first vertex in $\sigma$ that is not in $B$ and cannot reach $P$ in $G-B$.
We do a co-nondeterministic step so that in one branch of the computation we find a tree-partition for the connected component with representative $c$, and in the other branch we find the representative of the next connected component. The representative $c'$ of the next component is the first vertex in $\sigma$ such that it cannot reach a vertex appearing before $c$ (inclusive) in $\sigma$, nor a vertex of $P$ in $G-B$. When found, $c$ is replaced by $c'$ and we repeat this computation. If we don't find such a vertex $c'$, then $c$ must have represented the last connected component, so we simply accept.

Let us now describe what happens in the computation branch where we compute a new bag.
We can iterate on vertices in the component of $c$, by iterating on vertices of $G-B$ and then skipping if they are not reachable from $c$ in $G-B$. In particular, we can guess a subset $B'$ of size at most $k$ of vertices from this component. We then check that the neighborhood of $B$ in this component is contained in $B'$. If it is the case, we can set $P:=B$ and $B:=B'$ and recurse. If not, we reject.

If the graph is not connected, we can iterate on its connected components by using the same technique of remembering a vertex representative. For each of these components, we apply the above algorithm, with the modification that in each enumeration of the vertices we skip the vertex if it is not contained in the current component.

During these computations, we store at most $3k+O(1)$ vertices and use log-space subroutines.
Furthermore, the described computation tree is of polynomial size.
\end{proof}

The following problem is shown to be XALP-complete in \cite{XALP} and is the starting point of our reduction.

\defparaproblem{\textsc{Tree-Chained Multicolored Independent Set}}{A tree $T=(V_T,E_T)$, an integer $k$, and for each $i \in V_T$, a collection of $k$ pairwise disjoint sets of vertices $V_{i,1},\ldots,V_{i,k}$ and a graph $G$ with vertex set $V=\bigcup_{i \in V_T, j \in [1,k]} V_{i,j}$}{$k$}{Is there a set of vertices $W \subset V$, such that $W$ contains exactly one vertex from each $V_{i,j}$ ($i \in V_T, j \in [1,k]$), and for each pair $V_{i,j},V_{i',j'}$ with $i=i'$ or $ii' \in E_T$, $j,j' \in [1,k]$, $(i,j)\neq(i',j')$, the vertex in $W \cap V_{i,j}$ is non-adjacent to the vertex in $W \cap V_{i',j'}$?}

We remark that in the definition above, $V_{i,j}\cap V_{i',j'}=\emptyset$ whenever $(i,j)\neq (i',j')$.
We further use that we can assume the tree $T$ to be binary without loss of generality (see \cite{XALP} for more details). See Figure \ref{fig:tcmis} for a graphical representation of how the instance is arranged locally.

\begin{figure}[ht]
    \centering
    \includegraphics{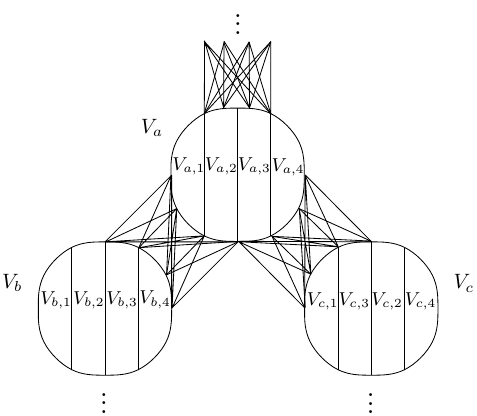}
    \caption{Local structure of \textsc{Tree-Chained Multicolored Independent Set}. For each $ab \in E_T$, the subgraph of $G$ induced by $V_a \cup V_b$ is a \textsc{Multicolored Independent Set} instance.}
    \label{fig:tcmis}
\end{figure}

We first give a brief sketch of the structure of the hardness proof. We have a \emph{trunk} gadget to enforce the shape of the tree from the \textsc{Tree-Chained Multicolored Independent Set}. On the trunk are attached \emph{clique chains} which are longer than the part of trunk between their endpoints, and have some wider parts at some specific positions. The length of the chain gives us some slack which will be used to encode the choice of a vertex for some subset $V_{i,j}$. Based on the edges of $G$, we adjust the width along the trunk so that only one clique chain may place its wider part on each position of the trunk. In other words, part of the trunk is a collection of gadgets representing edges of $G$ that allow for only one incident vertex to be chosen. See Figure \ref{fig:cliquechains} for a high level graphical representation of the gadgets.

\begin{figure}[ht]
    \centering
    \includegraphics[scale=0.3]{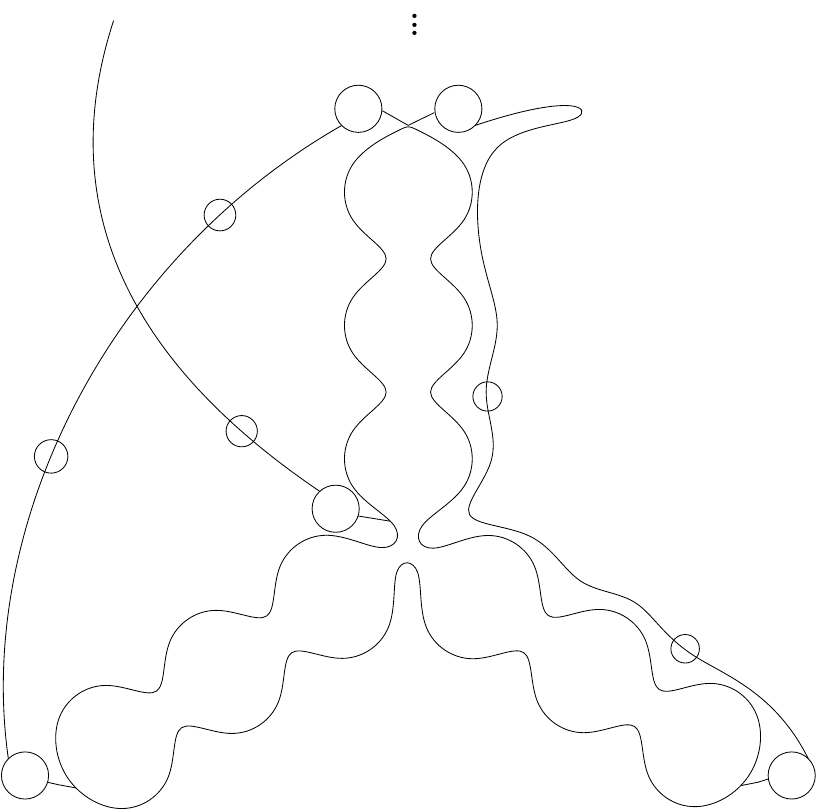}
    \caption{Local structure of our construction : a trunk with attached clique chains. All clique chains have to be folded at their endpoints to fit on the trunk as illustrated by the clique chain on the right.}
    \label{fig:cliquechains}
\end{figure}

\begin{lemma}
{\sc Tree-Partition-Width} with target width and maximum degree as combined parameter is XALP-hard.
\label{lemma:tpwxalphard}
\end{lemma}

\begin{proof}
We reduce from {\sc Tree-chained Multicolored Independent Set}.

Suppose that we are given a binary tree $T=(V_T,E_T)$, and for each node $i\in V_T$, a $k$-colored vertex set $V_i$.
We denote the colors by integers in $[1,k]$, and write $V_{i,j}$ for the set of vertices in $V_i$ with color $j$.
We are also given a set of edges $E$ of size $m$. Each edge in $E$ is a pair of vertices in $V_i \times V_{i'}$ with $i=i'$ or
$ii'$ an edge in $E_T$. We can assume the edges are numbered: $E= \{e_1, \ldots, e_m\}$.

In the {\sc Tree-chained Multicolored Independent Set} problem, we want to choose one vertex from each set $V_{i,j}$, $i\in V_T$,
$j\in [1,k]$, such that for each edge $ii'\in E_T$, the chosen vertices in $V_i \cup V_{i'}$ form an independent set (which thus
will be of size $2k$).

We assume that each set $V_{i,j}$ is of size $r$. (If not, we can add vertices adjacent to
all other vertices in $V_{i,j'}$, for all $j'\in [1,k]$. Such vertices cannot be in the solution.)
Write $V_{i,j} = \{v_{i,j,1}, v_{i,j,2}, \ldots, v_{i,j,r}\}$.

Let $N= 2(m+1)r$.
Let $L = 84k+5$.

\paragraph*{Cluster Gadgets}
In the construction, we use a {\em cluster gadget}. Suppose $Z$ is a clique. Adding a
cluster gadget for $Z$ is the following operation on the graph that is constructed.
Add a clique with vertex set $C_Z = \{c_{Z,1}, c_{Z,2}, \ldots, c_{Z,2L}\}$ of size $2L$ to the graph, and add an edge between each vertex in $Z$ and each vertex $c_{Z,j}$, $1\leq j\leq L$, i.e,
$Z$ with the first $L$ vertices in $C_Z$ forms a clique.

In a tree partition of a graph, the vertices of a clique can belong to at most two different bags.
The cluster gadget ensures that the vertices of clique $Z$ belong to exactly one bag. This cluster
gadget will be used in two different steps in the construction of the reduction.

\begin{claim}
Suppose a graph $H$ contains a clique $Z$ with the cluster gadget for $Z$. In each tree partition
of $H$ of width at most $L$, there is a bag that contains all vertices from $Z$.
\end{claim}

\begin{proof}
There must be two adjacent bags that contain the vertices of $C_Z$ and no other vertices.
Similarly, there must be two adjacent bags containing all vertices in $Z\cup \{c_{Z,1},\dots,c_{Z,L}\}$. This forces all vertices in $\{c_{Z,1},\dots,c_{Z,L}\}$ to be in a single bag, and all vertices in $Z$ to be in a single adjacent bag.
\end{proof}

\paragraph*{A subdivision of $T$}
The first step in the construction is to build a tree
$T'=(V_{T'},E_{T'})$, as follows. Choose an arbitrary node $i$ from $V_T$. Add a new neighbor $i'$ to $i$,
Add a new neighbor $r_0$ to $i'$. Now subdivide each edge $N$ times. The resulting tree is $T'=(V_{T'},E_{T'})$. We view $T'$ as a rooted
tree, with root $r_0$. We will use the word \emph{grandparent} to refer to the parent of the parent of a node. The nodes that do not result from the subdivisions are referred to as \emph{original nodes} (i.e. $V_T \cup \{i',r_0\}$ is the set of original nodes).
Nodes $i\in V_T$ and their copies in $T'$ will be denoted with the same name.

The graph $H$ consists of two main parts: the \emph{trunk} and the \emph{clique chains}. To several cliques in these parts,
we add cluster gadgets.

\paragraph*{The trunk}
The trunk is obtained by taking for each node $i\in V_{T'}$ a clique $A_i$. We specify below the size of these cliques.
For each edge $ii'$ in $T'$, we add an edge between each vertex in $A_i$ to each vertex in $A_{i'}$.
We add for each $A_i$ a cluster gadget.

To specify the sizes of sets $A_i$, we first need to give some definitions:
\begin{itemize}
    \item For each node $i'\in V_{T'}$, we let $p(i')$ be the number of nodes $i \in V_T$ (i.e., `original nodes'), such that
$i'$ is on the path (including endpoints) in $T'$ from $i$ to the vertex that is the grandparent of $i$ in $T$. I.e., for each original node $i$,
we look to the grandparent of $i$ (if it exists), and then add 1 to the count of each node $i'$ on the path between them in $T'$.
    \item For each edge $e_j\in E$, let $g(e_j) = 2jr$.
    \item For each edge $e_j = \{v_{i,c,s}, v_{i',c',s'}\}$, we have that $i=i'$ or $i'$ is a child of $i$. Let $i_{e_j}$ be the node in $T'$, obtained by making $g(e_j)$ steps up in $T'$ from $i$: i.e., $i_{e_j}$ is the ancestor of $i$ with distance $g(e_j)$ in $T'$.
\end{itemize}

Now, for all nodes $i\in V_{T'}$,
\begin{itemize}
    \item $|A_i|$ equals $L - 6 k \cdot p(i) - 1$, if $i= i_{e_j}$ for some $e_j\in E$. At this node, we will verify
    that a choice (encoded by the clique chains, explained below), indeed gives an independent set: we check that we did not choose both endpoints of $e_j$.
    \item $|A_i|$ equals $L - 6 k \cdot p(i) - 2$, otherwise.
\end{itemize}

\paragraph*{The clique chains}
For each $i\in V_T$, and each color class $c\in [1,k]$, we have a clique chain with
$2 N + r + 5$ cliques, denoted $CC_{i,c,\gamma}$, $\gamma\in [1,2N+r+5]$.
All vertices in the first clique $CC_{i,c,1}$ are made incident to all vertices in $A_i$.
All vertices in the last clique $CC_{i,c,2N+r+5}$ are made incident to all vertices in $A_{i'}$ with $i'$ the parent of
the parent (i.e., the grandparent) of $i$ in $T$. (Notice that the distance from $i$ to $i'$ in $T'$ equals $2(N+1)$.)
All vertices in $CC_{i,c,\gamma}$ are made incident to all vertices in $CC_{i,c,\gamma+1}$, i.e., all vertices in a clique are
adjacent to all vertices in the next clique in the chain.

To each clique $CC_{i,c,\gamma}$ we add a cluster gadget.

The cliques have different sizes, which we now specify. Informally, the first and last clique are large enough to enforce the way they attach to the trunk, every other clique is of constant size, with some cliques being slightly larger to enforce constraints corresponding to edges of $E$.
More precisely, the edges within some $V_i$ are checked in the first half of a chain, while the edges between $V_a$ and $V_{b}$ with $a$ the parent of $b$ are checked in the first half of chains encoding at $V_a$ and the second half of the chains encoding $V_b$ (these parts of the chains share the same section of the trunk).
Consider $i\in V_T$, $c\in [1,k]$, $\gamma \in [1, 2N+r+5]$.
 The size of $CC_{i,c,\gamma}$ equals:
\begin{itemize}
    \item $L-7$, if $\gamma=1$ or $\gamma= 2N+r+5$ (i.e., for the first and last clique in the chain.)
    \item 7, if there is an edge $e_j \in E$ with one endpoint in $V_{i,c}$ for which one of the following cases holds:
    \begin{itemize}
        \item $e_j = \{v_{i,c,\alpha}, v_{i,c',\alpha'}\}$, $c\neq c'$, i.e., one endpoint is in $V_{i,c}$, and the other endpoint
        is in another color class in $V_i$, and $\gamma = g(e_j)+1+\alpha$.
        \item $e_j = \{v_{i,c,\alpha}, v_{i',c',\alpha'}\}$, $i'$ is a child of $i$, and
        $\gamma = g(e_j)+1+\alpha$.
        \item $e_j = \{v_{i,c,\alpha}, v_{i',c',\alpha'}\}$, $i$ is a child of $i'$, and
        $\gamma = g(e_j)+N+2+\alpha $.
    \end{itemize}
    \item 6, otherwise
\end{itemize}

\begin{figure}[ht]
    \centering
    \includegraphics{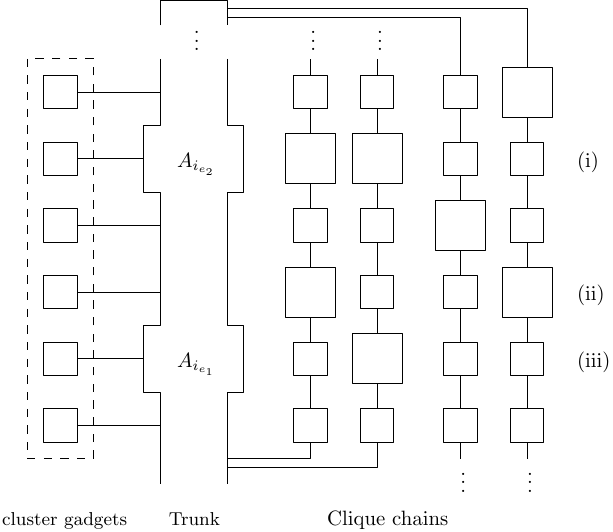}
    \caption{Local structure of the gadgets. (i) If two larger cliques from the clique chains are aligned with the edge constraint of the trunk (encoding that we chose the two endpoints of the edge), then the width is slightly too large locally. (ii) If two larger cliques from clique chains are aligned but not aligned with an edge constraint, this fits the expected width (there cannot be more than 2). (iii) If at most one larger clique from clique chains is aligned with an edge constraint of the trunk, it fits the expected width.}
    \label{fig:gadgets}
\end{figure}

Let $H$ be the resulting graph (see Figures \ref{fig:cliquechains} and \ref{fig:gadgets}).

\begin{claim}
$H$ has tree-partition-width at most $L$, if and only if the given instance of {\sc Tree-chained Multicolored Independent Set}
has a solution.
\end{claim}

\begin{proof}
Suppose we have a solution of the {\sc Tree-chained Multicolored Independent Set}. We define $h(i,c)$ for
each class $V_{i,c}$, such that the vertex $v_{i,c,h(i,c)}$ is the vertex chosen by the solution in this class. Now, we can construct the tree partition as follows.
First, we take the tree $T'$, and for each node $i$ in $T'$, we take a bag initially containing the vertices in $A_i$;
we later add more vertices to these bags in the construction.

Now, we add the chains, one by one. For a chain for $V_{i,c}$, take a new bag that contains $CC_{i,c,1}$, and make
this bag adjacent to the bag of $i$. We add the vertices of $CC_{i,c,h(i,c)+1}$ to the bag of $i$. If $h(i,c)>1$, then
we place the vertices of cliques $CC_{i,c,\alpha+1}$ with $1<\alpha<h(i,c)$ in bags outside the trunk: $CC_{i,c,h(i,c)}$
goes to the bag with $CC_{i,c,1}$; to this bag, we add an adjacent bag with $CC_{i,c,2} \cup CC_{i,c,h(i,c)-1}$; to this, we add an adjacent bag with $CC_{i,c,3} \cup CC_{i,c,h(i,c)-2}$, etc.

Now, add the vertices of $CC_{i,c,h(i,c)+2}$ to the bag of the parent of $i$, and continue this: each next clique
is added to the next parent bag, until we add a clique to the bag of grandparent of $i$ in $T$; name this
node here $i''$. Add a new bag incident to $i''$ and put $CC_{i,c,2N+r+5}$ in this bag (i.e., the last clique of the chain). Similar as at the start of the chain, fold the end of the chain (with possibly some additional new bags) such that a bag containing $CC_{i,c,2N+r+4}$ is adjacent to the bag with $CC_{i,c,2N+r+5}$.

Finally, for each cluster gadget, add two new bags, with the first incident to the bag containing the respective clique.

One easily verifies that this gives a tree partition of $H$. For bags outside the trunk, one easily observes that
the size is at most $L$. Bags $i$ in the trunk contain a set $A_i$, and precisely $p(i)\cdot k$ cliques of the clique chains:
for each path that counts for the bag, and each color class in $[1,k]$, we have one chain with one clique.
Each of these cliques has size six or seven. Now, we can notice that a clique of size 7 corresponds to an edge $e_{j'}$ with endpoint in the class. This clique will be mapped to a node in the trunk that
equals $i_{e_{j'}}$, if and only if this endpoint is chosen; otherwise, the clique will be mapped to
a trunk node with distance less than $r$ to $i_{e_{j'}}$. Thus, there are two cases for a trunk node $i$:
\begin{itemize}
    \item There is no edge $e_j$ with $i=i_{e_j}$. Then, a close observation of the clique chains shows that there are
    at most two clique chains with size 7 mapped to $i$. Indeed, the construction is such that each edge has its private interval, and affects the trunk both between $i$ and its parent $i'$, and between $i'$ and its parent $i''$.
    \item $i=i_{e_j}$. Now, at most one endpoint of $e_j$ is in the solution. The clique chain of the color class of that endpoint can have a clique of size 7 mapped to $i$. The `offset' of the clique chain for the color class of the other
    endpoint is such that there is a clique of size 6 for that chain at $i$.
\end{itemize}
In both cases, the total size of the bag at $i$ is at most $L$. Thus, the width of the tree partition is at most $L$.

\smallskip

Suppose we have a tree partition of $H$ of width $L$.
First, by the use of the cluster gadgets, each clique $A_i$ is in one bag. A bag cannot contain two cliques $A_i$ as
each has size larger than $L/2$. Now, the bags containing $A_i$ form a subtree of the partition tree that is
isomorphic to $T'$. For each clique chain of a class $V_{i,c}$, we have that the first clique $CC_{i,c,1}$
is in a bag adjacent to $i$,
and the last clique $CC_{i,c,2N+r+5}$ is in a bag adjacent to the trunk bag that corresponds to
the grandparent of $i$ in $T$, say $i''$. Each trunk bag from $i$ to $i''$ thus must contain
a clique (of size 6 or 7) from the clique chain of $V_{i,c}$. It follows that each trunk bag
$i'$ contains at least $p(i')\cdot k$ cliques of size at least 6 each of the clique chains.
Now, $|A_{i'}|+ 6p(i')\cdot k \in \{L-1,L-2\}$, and thus we cannot add another clique of a clique
chain to a trunk bag.

For a clique chain of $V_{i,c}$, there is a clique mapped to the trunk bag of $i$.
Suppose $CC_{i,c,h(i,c)}$ is mapped to $i$. We claim that choosing from each $V_{i,c}$
the vertex $v_{i,c,h(i,c)}$ gives an independent set, and thus, we have a solution of
the {\sc Tree Chained Multicolored Independent Set} problem.

The vertices of $CC_{i,c,h(i,c)+2}$ must be mapped to the bag of the parent of $i$, as otherwise,
$i$ will contain an additional clique of size at least 6, and the size of the bag of $i$ will become larger than $L$. By induction, we have that the $\alpha$th parent of $i$, $\alpha\in [1,2N+2]$
contains the vertices of $CC_{i,c,h(i,c)+\alpha+1}$. (Note that the $(2N+2)$nd parent equals the
node corresponding the grandparent of $i$ in $T$.)

We now consider the node $i_{e_j}$ for edge $e_j\in E$. Suppose $e_j = v_{i,c,\alpha}v_{i',c',\alpha'}$.
Without loss of generality, suppose $i=i'$ or $i'$ is a child of $i$; otherwise, switch roles of $i$ and $i'$.
For each endpoint of this edge, if the endpoint is chosen (i.e., $\alpha = h(i,c)$ or $\alpha'= h(i',c')$),
then the corresponding clique chain has a clique of size 7 in the bag $i_{e_j}$.
This can be seen by the following case analysis:
\begin{itemize}
    \item By assumption, $CC_{i,c,h(i,c)+1}$ is placed in the bag of $i$. As each successive clique in the chain is placed in one higher bag along the path from $i$ to the grandparent of $i$ (in $T$), we have that
    $CC_{i,c,h(i,c)+g(e_j)+1}$ is placed in the bag of $i_{e_j}$, as this node is the $g(e_j)$th parent of $i$ in $T'$. This clique has size 7.
    \item If $i'=i$, the same argument shows that $CC_{i',c',h(i',c')+1}$ is a clique of size 7 placed in the bag of $i_{e_j}$.
    \item If $i'$ is a child of $i$ in $T$, then $CC_{i',c',h(i',c')+1}$ is placed in the bag of $i'$.
    Again, each successive clique in the chain of $V_{i',c'}$ is placed in the next parent bag, for all
    nodes on the path from $i'$ to the grandparent of $i'$ in $T$ (which is the parent of $i$ in $T$.)
    This implies that $CC_{i',c',h(i',c')+N+2}$ is placed in the bag of $i$ and
    $CC_{i',c',h(i',c')+N+2+g(e_j)}$ is placed in the bag of $i_{e_j}$; this bag has size 7.
\end{itemize}
Now, if both endpoints were to be chosen, then the size of the bag of $i_{e_j}$ would be larger than $L$: it
contains $A_{i_{e_j}}$ (which has size $L-6kp(i_{e_j})-1$), $kp(i_{e_j})$ cliques from clique chains, of which
all have size at least 6 and two have size 7; contradiction. So, at most one endpoint is chosen, so
choosing vertices $v_{i,c,h(i,c)}$ gives an independent set.
\end{proof}

The maximum degree of a vertex in $H$ is less than $5kL+5L = O(k^2)$:
\begin{itemize}
    \item Vertices in cluster gadgets have maximum degree less than $2L$.
    \item A vertex in a trunk clique $A_i$ of a node $i$ that resulted from a subdivision has maximum degree less than $4L$ as $i$ has two incident nodes, each with a trunk clique of size less than $L$, and there is a cluster gadget attached to $A_i$.
    \item A vertex in a trunk clique $A_i$ of a node $i$ that is an original node (i.e., also in $T$) has
    less than $L$ neighbors in $A_i$, less than $3L$ neighbors in $A_{i'}$ with $i'$ incident to $i$ in $T'$,
    less than $kL$ neighbors of cliques $CC_{i,c,1}$ (one clique of size $L-7$ for each class $c\in [1,k]$),
    less than $4kL$ neighbors of cliques $CC_{i',c,2N+r+5}$ (one clique of size $L-7$ for each
    node of which $i$ is the grandparent in $T$ for each class $c\in [1,k]$), and less than $L$ neighbors in the cluster gadget attached to $A_i$.
\end{itemize}

Finally, we conclude that the transformation can be carried out in $f(k)\log n$ space, thus the result follows.
\end{proof}

From Lemmas~\ref{lemma:tpwinxalp} and Lemma~\ref{lemma:tpwxalphard}, the following result directly follows.

\begin{theorem}
{\sc Tree-Partition-Width} is XALP-complete, when the target width is the parameter, or when the target width plus the maximum degree is the parameter.
\end{theorem}

\begin{theorem}
{\sc Domino Treewidth} is XALP-complete.
\end{theorem}

\begin{proof}
We first prove membership in XALP.

We use the fact that the maximum degree of the graph is bounded by $2k$ where $k$ is the domino treewidth. We can discard an instance where this condition on the maximum degree is not satisfied in log-space. We first assume that the given graph is connected.

The ``certificate'' used for this computation will be of size $O(k^2+k\log(n))$ and consists of:
\begin{itemize}
    \item The current bag and for each of its vertices whether it was contained in a previous bag or not.
    This requires at most $k + k\log(n)$ bits.
    \item For each neighbor of the bag, whether it was already covered by a bag. This requires $O(k^2)$ bits (We use some fixed ordering on the vertices to decode which bit corresponds to which vertex in log-space).
\end{itemize}

The algorithm works as follows.
Given the current certificate, if all neighbors have been covered we accept.
Otherwise, we guess a new child bag by picking a non-empty subset of $k+1$ vertices among the vertices of the current bag that were contained only in this bag, and the neighbors that were not already covered. We then check that each vertex that is in both the current bag and child bag has all of its neighbors in these two bags. We then guess for each not already covered neighbor of the current bag if it should be covered by the subtree of this child. These vertices are then considered as covered in the current bag certificate. In the child bag certificate, the non covered neighbors are these vertices and the neighbors of the child bag that are not neighbors of the parent bag.
We then recurse with both certificates, and accept if both recursions accept.

This computation uses $O(k^2 + k\log n)$ space and the computation tree has polynomial size.

We can handle disconnected graphs by iterating on component representatives and discarding vertices that are not reachable using the fact that reachability in undirected graphs can be computed in log-space (see membership for \textsc{Tree-Partition-Width} for more details).

Hardness follows from a reduction from {\sc Tree-Partition-Width} when we use the target value and maximum degree as parameter.

Suppose we are given a graph $G=(V,E)$ of maximum degree $d$ and an integer $k$.

Let $L= kd+1$, and $M =(k+1)L -1$. Now, build a graph $H$ as follows.
For each vertex $v\in V$, we take a clique $C_v$ with $L$ vertices.
For each vertex $w\in C_v$, we add a set $S_w$ with $2M-2$ vertices, and make one of the vertices in $S_w$ incident to
$w$ and to all other vertices in $S_w$; call this vertex $y_w$.

For each edge $e=vw\in E$, we add a vertex $z_e$, and make $z_e$ incident to all vertices in $C_v$ and all vertices
in $C_w$. Let $H$ be the resulting graph.

\begin{claim}
$G$ has tree-partition-width at most $k$, if and only if $H$ has domino treewidth at most $M-1$.
\end{claim}

\begin{proof}
Suppose $H$ has domino treewidth at most $M-1$. Suppose $(\{X_i|i\in I\}, T=(I,F))$ is a domino tree decomposition of $H$
of width at most $M-1$, i.e., each bag has size at most $M$.

First, consider a vertex $w$ in some $C_v$. The vertex $y_w$ has degree $2M-2$, which implies that there are two adjacent bags
that each contain $y_w$, and $M-1$ neighbors of $y_w$. One of these bags contains $w$.

For each $v\in V$, there must be at least one bag that contains all vertices of $C_v$, by a well known property of
tree decompositions. There can be also at most one such bag, because each vertex $w\in C_v$ is in another bag
that is filled by $w$, $y_w$, and $M-2$ other neighbors of $y_w$.

For each $i\in I$, let $Y_i \subseteq V$ be the set of vertices $v\in V$ with $C_v \subseteq X_i$. We claim that
$(\{Y_i|i\in I\}, T=(I,F))$ is a tree partition of $G$ of width at most $k$ (some bags are empty). First, by the discussion above,
each vertex $v\in V$ belongs to exactly one bag $Y_i$. Second, as $M < (k+1)L$, each $Y_i$ has size at most $k$.
Third, if we have an edge $e=vw \in E$, then $z_e$ is in the bag that contains $C_v$, and $z_e$ is in the bag
that contains $C_w$. As $z_e$ is in at most two bags, these two bags must be the same, or adjacent, so
in $(\{Y_i|i\in I\}, T=(I,F))$, $v$ and $w$ are in the same set $Y_i$ or in sets $Y_i$ and $Y_{i'}$ with $i$ and $i'$ adjacent in the decomposition tree $T$.

\smallskip

Now, suppose $G$ has tree-partition-width at most $k$, say with tree partition $(\{Y_i|i\in I\}, T=(I,F))$.
For each $i\in I$, let $X_i = \bigcup_{v\in Y_i} C_v \cup \{ z_e ~|~ \exists v\in Y_i, w\in V: e = vw\}$.
For each $v\in V$, $w\in C_v$, add two bags, one containing $w$, $y_w$, and $M-2$ other neighbors of $y_w$, and the other
containing $y_w$ and the remaining $M-1$ neighbors of $y_w$, and make these bags adjacent, with the first adjacent to
the bag in $T$ that also contains $w$. One easily verifies that this results in a domino tree decomposition of $H$
with maximum bag size at most $M$, hence $H$ has domino treewidth at most $M-1$.
\end{proof}

It is easy to see that $H$ can be constructed from $G$ with $f(k)\log |V|$ memory. So, the hardness of
{\sc Domino Treewidth} follows from the previous lemma.
\end{proof}

\section{Tree-cut width and the stability of tree-partition-width}\label{section:tcw}
In this section, we consider the relation of the notion of tree-cut width
with (stable) tree-partition-width. Tree-cut width was introduced by Wollan~\cite{Wollan15}. Ganian et al.~\cite{Ganian0S15} showed that several problems that are $W[1]$-hard with treewidth as parameter are fixed parameter tractable with tree-cut width as parameter.

We begin by defining the \emph{tree-cut width} of a graph $G=(V,E)$.
A tree-cut decomposition $(T,\mathcal{X})$ consists of a rooted tree $T$ and a family $\mathcal{X}$ of \emph{bags} $(X_i)_{i \in V(T)}$ which form a near partition of $V(G)$ (i.e. some bags may be empty, but nonempty bags form a partition of $V(G)$). For $t \in V(T)$, we denote by $e(t)$ the edge of $T$ incident to $t$ and its parent. For $e \in E(T)$, let $T_1,T_2$ denote the two connected components of $T-e$. We denote by $\cut(e)$ the set of edges with an endpoint in both of $\bigcup_{i \in V(T_1)} X_i$ and $\bigcup_{i \in V(T_2)} X_i$. The adhesion of $t \in V(T)$ is $\text{adh}(t)=|\cut(e(t))|$, and its torso-size is $\text{tor}(t)=|X_t|+b_t$ where $b_t$ is the number of edges $e \in E(T)$ incident to $t$ such that $|\cut(e)| \geq 3$. The width of the decomposition is then $\max_{t \in V(T)} \{\text{adh}(t),\text{tor}(t)\}$. Note that edges are allowed to go between vertices that are not in the same bag.
The tree-cut width of a graph is the minimal width of tree-cut decomposition. When $|\cut(e)|\geq 3$, the edge $e$ is called \emph{bold}, and otherwise, $e$ is called \emph{thin}. When $\text{adh}(t)\leq 2$, node $t$ is called \emph{thin}, otherwise it is called \emph{bold}.
In \cite{Ganian0S15}, it is shown that a tree-cut decomposition can be assumed to be \emph{nice}, meaning that if $t \in V(T)$ is thin then $N(Y_t) \cap \left( \bigcup_{b \text{ sibling of }t} Y_b \right) = \varnothing$, where $Y_i$ is the union of $X_j$ for $j$ in the subtree of $i$.

Wollan~\cite{Wollan15}  shows that having bounded tree-cut width is equivalent to only having wall immersions of bounded size.

\begin{observation}\label{obs:tcwK3n}
    $\tpw(K_{3,n}) \leq 3$ and $\tcw(K_{3,n})=\Theta(\sqrt{n})$.
\end{observation}

\begin{proof}
    Let $(A,B)$ denote the bipartition of $K_{3,n}$ with $|A|=3$.

    A tree-partition of $K_{3,n}$ achieving width 3 is the partition with $A$ in one bag and every other vertex in a separate bag.

    It is easy to see that $K_{3,n}$ has a tree-cut decomposition of width $O(\sqrt{n})$: we place $A$ as the center of a star, with about $\sqrt{n}$ leaves of size $\sqrt{n}$.
    Now we consider an arbitrary tree-cut decomposition of $K_{3,n}$ achieving width $O(\sqrt{n})$. We first note that the vertices of $A$ cannot be split into separate bags because if they were, any edge of the decomposition on the path between such bags would have adhesion at least $n/2$.
    Hence, there is a bag containing $A$ and we may root the tree of our decomposition in this bag. Each subtree of the decomposition will contribute to the torso-size, and each vertex will contribute linearly to the adhesion of the edge from its subtree to the root. Since we assume the width to be $O(\sqrt{n})$, we must have at most $O(\sqrt{n})$ vertices per subtree. Consequently, there must be $\Omega(\sqrt{n})$ subtrees, so the torso-size of the root is $\Omega(\sqrt{n})$.
\end{proof}

Note that any graph on $n$ vertices with maximum degree $3$ can be immersed in $K_{3,n}$. In particular, this works for any wall. The lower bound given by Wollan shows that the tree-cut width of $K_{3,n}$ is $\Omega(n^{\frac{1}{4}})$.

We denote by $\underline{\tpw}(G)$ the minimum tree-partition-width over subdivisions of $G$ (stable tree-partition-width), and by $\overline{\tpw}(G)$ the maximum tree-partition-width of subdivisions of $G$.
We will show that both are polynomially tied to the tree-partition-width of $G$, which proves useful in polynomially bounding tree-partition-width by tree-cut width due to the following lemma.

\begin{lemma}\label{lem:tpw-tcw}
    $\underline{\tpw} = O(\tcw^2)$.
\end{lemma}

\begin{proof}
Consider a nice tree-cut decomposition $(T,\mathcal{X})$ of a graph $G$ of width $k$.
We will construct a tree-partition for a subdivision of $G$. Note that the bags are already disjoint, but some edges are not between neighboring bags of $T$.

Each edge $uv$ of $G$ is subdivided $d_T(u,v)$ times, which is the distance between the nodes containing $u$ and $v$ respectively in their bag (recall that the bags form a near partition). We then add the vertices of the subdivided edge in the bags on the path in the decomposition between the bags containing their endpoints.
This is sufficient to make the decomposition a tree-partition $T'$ of a subdivision of $G$.

We now argue that $T'$ has a width of $O(k^2)$. A bag $Y_t$ of $T'$ contains at most:
\begin{itemize}
    \item $k$ initial vertices
    \item $\max(2,k)$ vertices from subdivisions of edges in $\cut(e(t))$ accounting for edges going from a child of $t$ to an ancestor of $t$
    \item $k^2/2$ vertices from edges that are between bold children of $t$. For $u,v$ children of $T$, there are only edges between $T_u$ and $T_v$ if both are bold. There are at most $k$ bold children in a tree-cut decomposition. There are also at most $k$ such edges incident to $T_u$ for any child $u$ of $T$, and we may divide by 2 since each edge will be counted twice this way. We stress that thin children do not contribute because the tree-cut decomposition is nice.
\end{itemize}

Hence, $\underline{\tpw}(G) \leq k + (k+2) + k^2/2 = O(\tcw(G)^2)$
\end{proof}

Next, we consider the parameters $\tpw$ and $\overline{\tpw}$.

\begin{lemma}\label{lem:tpw-fold-subdivisions}
    $\tpw \leq \overline{\tpw} \leq \tpw(\tpw+1)$
\end{lemma}

\begin{proof}
The lower bound is immediate. We prove the upper bound.

Consider a graph $G$ with a tree-partition $(T,\mathcal{X})$ of width $k$, and a subdivision $G'$ of $G$.
We construct a tree-partition $(T',\mathcal{X}')$ of $G'$ of width at most $k(k+1)$.

We root the decomposition $T$ arbitrarily.

Suppose that $u,v$ are in the same bag of $T$ and the edge $uv$ was subdivided to form the path $u, a_1, \dots,$ $a_\ell, v.$ We add the vertices $a_i$ in new bags containing, $\{a_1,a_\ell\}$, $\{a_2,a_{\ell-1}\},\dots$ which corresponds to a new branch of the decomposition of width at most $2$.

Consider next the vertices obtained by subdividing an edge $uv$ for $u$ in the child bag of the bag of $v$.
If a subdivided edge was between two vertices of adjacent bags, we order the vertices of the path obtained by subdividing the edge from the vertex in the child bag to the vertex in the parent bag. We add the penultimate vertex to the child bag, and fold the remaining vertices of the path in a fresh branch of the decomposition of width at most 2 similarly to the previous case.

This gives a tree partition $T'$.
Bags of $T'$ that are not in $T$ have size at most $2$, and, to bags of $T'$ that are also in $T$, we added at most $k^2$ vertices (at most one per edge between the bag and its parent). We conclude that $T'$ has width at most $k(k+1)$.
\end{proof}

A result of Ding and Oporowski \cite{DingO96} shows that tree-partition-width is tied to a parameter $\gamma$ that is (by design) monotonic with respect to the topological minor relation.
We adapt their proof to derive the following stronger result.

\begin{theorem}\label{thm:tpw-obstructions}
    There exists a parameter $\gamma$ which is polynomially tied to the tree-partition-width, and is monotonic with respect to the topological minor relation.
    More precisely, $\tpw = \Omega(\gamma)$ and $\tpw = O(\gamma^{24})$.
\end{theorem}

We deduce the following statement.

\begin{corollary}\label{coro:stable-tpw}
    $\tpw$, $\overline{\tpw}$, and $\underline{\tpw}$ are polynomially tied.
\end{corollary}

\begin{proof}
    Lemma~\ref{lem:tpw-fold-subdivisions} shows that $\tpw$ and $\overline{\tpw}$ are polynomially tied.
    Note that, by definition, $\underline{\tpw} \leq \tpw$.
    Then, for a fixed graph $G$, consider $H$ a subdivision of $G$ achieving $\tpw(H)=\underline{\tpw}(G)$. Then
    $\tpw(G)^{O(1)} \leq \gamma(G) \leq \gamma(H) \leq \tpw(H)^{O(1)}=\underline{\tpw}(G)^{O(1)}$. The first and last inequalities come from the fact that $\gamma$ and $\tpw$ are polynomially tied. The middle inequality is because $\gamma$ is monotonic with respect to the topological minor relation.
\end{proof}

From Lemma~\ref{lem:tpw-tcw} and Corollary~\ref{coro:stable-tpw}, we deduce the following theorem.

\begin{theorem}\label{thm:tpw-tcw}
    The parameter tree-partition-width is polynomially bounded by the parameter tree-cut width. In other words, we show that there exist constants $C,c>0$ such that for any graph $G$, $\tpw(G)\leq C\tcw(G)^c$.
\end{theorem}

We now turn our focus to the technical proof of Theorem~\ref{thm:tpw-obstructions}, for which we first need some further definitions and results.

We define the $m$-grid as the graph on the vertex set $[m]\times[m]$ with edges $(i,j)(i',j')$ when $|i-i'| + |j-j'| \leq 1$. We then define the $m$-wall as the graph obtained from the $m$-grid by removing edges $(i,j)(i+1,j)$ for $i+j$ even. The \emph{wall number} of a graph $G$ is then defined as the largest $m$ such that $G$ contains the $m$-wall as a (topological)\footnote{The notions of minor and topological minor coincide for graphs of maximum degree at most 3.} minor, and the \emph{grid number} of $G$ is the largest $m$ such that $G$ contains the $m$-grid as a minor. We denote them by $\wn(G)$ and $\gn(G)$ respectively.

\begin{observation}
    The wall number and the grid number are linearly tied: $\wn(G)=\Theta(\gn(G))$.
\end{observation}

We use the following result of Chuzhoy and Tan \cite{ChuzhoyT21} (the bound is weakened to have a lighter formula).

\begin{lemma}[Chuzhoy and Tan~\cite{ChuzhoyT21}]\label{lem:wn}
    The treewidth is polynomially tied to the grid number: $\tw = \Omega(\gn)$ and $\tw = O(\gn^{10})$.

    Hence, the treewidth is polynomially tied to the wall number: $\tw = \Omega(\wn)$ and $\tw = O(\wn^{10})$.
\end{lemma}
We are now ready to prove the theorem.

\begin{proof}[of Theorem~\ref{thm:tpw-obstructions}]
We call \emph{$m$-fan} the graph that consists of a path of order $m$ with an additional vertex adjacent to all of the vertices of the path.
We call \emph{$m$-branching-fans} the graphs that consist of a tree $T$ and a vertex $v$ adjacent to a subset $N$ of the vertices of $T$ containing at least the leaves, such that $m$ is the minimum size of a subset of vertices $X$ of $T$ such that each component of $T-X$ contains at most $m$ vertices of $N$.
In particular, the $(m+1)^2$-fan is an $(m+1)$-branching-fan.
We call \emph{$m$-multiple of a tree of order $m$} a graph obtained from a tree of order $m$ after replacing its edges by $m$ parallel edges and then subdividing each edge once to keep the graph simple.

Let $\gamma_1(G)$ be the largest $m$ such that $G$ contains an $m$-branching-fan as a topological minor.
Let $\gamma_2(G)$ be the largest $m$ such that $G$ contains an $m$-multiple of a tree of order $m$ as a topological minor.

Let $\gamma(G)$ be the maximum of $\wn(G)$, $\gamma_1(G)$, and $\gamma_2(G)$.
\begin{claim}
    The parameter $\gamma$ is monotonic with respect to the topological minor relation.
\end{claim}

\begin{proof}
    Let $G$ be a graph and $H$ be a topological minor of $G$. Any topological minor of $H$ is also a topological minor of $G$, hence $\wn(G) \geq \wn(H)$, $\gamma_1(G) \geq \gamma_1(H)$, $\gamma_2(G) \geq \gamma_2(H)$.
    We conclude that $\gamma(G) \geq \gamma(H)$.
\end{proof}

We remark that the $m$-branching-fans, the $m$-multiples of trees of order $m$ and the $m$-wall have tree-partition-width $\Omega(m)$. Hence, we have $\tpw(G) = \Omega(\gamma(G))$.

We fix a graph $G$ and let $m=\gamma(G)$. Note that $m \geq \gamma_2(G) >0$.

We denote by $G^b$ the graph on the vertex set of $G$, where $xy$ is an edge if and only if there are at least $b$ vertex disjoint paths from $x$ to $y$. We now consider $G^b$ for $b=\Omega(m^{10})$.

\begin{claim}\label{claim:small-comp}
    The connected components of $G^b$ have size at most $m$.
\end{claim}

\begin{proof}
    We proceed by contradiction, and assume there is a connected component $C$ of size at least $m+1$.

    Since $C$ is connected, it contains a spanning tree $T$. We number its edges $e_1,\dots,e_\ell$ such that every prefix induces a connected subtree of $T$.
    We construct a subgraph $H$ of $G$ that should be an $(m+1)$-multiple of a tree of order $m+1$, contradicting the definition of $m$.
    For each edge $uv$, in order, we try to add to $H$ $m+1$ vertex disjoint paths from $u$ to $v$ that avoid vertices of $C$ and the vertices already in $H$.
    If we manage to do this for at most $m$ edges, then we have placed at most $m(m+1)$ paths.
    Let $uv$ be the first edge for which we could not find $m+1$ vertex disjoint paths that do not intersect previous vertices (except for $u$ or $v$). By definition of $G^b$, there are $b$
    vertex disjoint $u,v$-paths in $G$, we denote the set of such paths by $\pi$. At most $m$ of the paths of $\pi$ hit vertices of $C$ already in $H$. Then, since at least $b-m$ are hit by previous paths and there are at most $m(m+1)$ previous paths. By the pigeon hole principle, one of the previous paths $P_0$ must hit $\frac{b-m}{m(m+1)} \geq (m+1)^2$ paths in $\pi$. By considering $P_0$ and the paths it hits in $\pi$, we easily obtain a subdivision of an $(m+1)^2$-fan. This is a contradiction with the definition of $m$.
    Hence, we must have been able to process $m$ edges. Which means we obtained a subdivision of an $(m+1)$-multiple of a tree of order $m+1$. This is a contradiction to the definition of $m$.
    We conclude that the connected component must have size at most $m$.
\end{proof}

Let $H$ be the quotient of $G$ by the connected components of $G^b$. We call it the $b$-reduction of $G$.

\begin{claim}\label{claim:deg-red}
    The blocks of $H$ have maximum degree at most $bm^3$.
\end{claim}

\begin{proof}
    Assume by contradiction that the maximum degree is more than $bm^3$.
    Let $B$ be a block of $H$, and $X$ be one if its vertices of maximum degree.
    $X$ contains at most $m$ vertices of $G$ by Claim~\ref{claim:small-comp}. The vertices of $G$ in $B-X$ must be in the same connected component $C$ of $G$ since $b>m$.
    There are at least $bm^3 +1$ edges between $X$ and $C$. By the pigeon hole principle, one vertex $v$ of $X$ must have at least $bm^2 + 1$ neighbors in $C$. Consider a spanning tree $T$ of $C$. We iteratively remove leaves that are not neighbors of $v$, and then replace any vertex of degree $2$ that is not a neighbor of $v$ by an edge between its neighbors. We denote this reduced tree by $T'$.

    First, note that the degree in $T'$ is bounded by $b-1$ because incident edges can be extended to vertex disjoint paths to leaves of $T'$ which are neighbors of $v$ by construction. We now use the fact that $G$ contains no $(m+1)$-branching-fans as topological minors. In particular, there must be a set $U$ of vertices of $T'$ of size at most $m$ such that components of $T'-U$ contain at most $m$ neighbors of $v$. By removing at most $m$ vertices of degree at most $(b-1)$, we have at most $1+(b-2)m$ components in $T'-U$ meaning $v$ has degree bounded by $(b-1)m^2$. We found our contradiction.
\end{proof}

\begin{claim}\label{claim:tw-red}
    The treewidth of $H$ is at most $O(b)$.
\end{claim}

\begin{proof}
    We first apply Lemma~\ref{lem:wn} to bound the treewidth of $G$ by $O(b)$.
    Consider a tree decomposition of $G$ of adequate width $\Theta(b)$, and replace each bag by the components of $G^b$ that intersect it. By Claim~\ref{claim:dec-red}, this is a decomposition of $H$.
\end{proof}

Using Claim~\ref{claim:deg-red} and Claim~\ref{claim:tw-red} and the construction of Wood as we did in the approximation algorithm, we obtain a tree-partition of $H$ of width $O(b^2m^3)$. We then replace components of $G^b$ by their vertices, obtaining a tree-partition of $G$ of width $O(b^2m^4)$ due to Claim~\ref{claim:small-comp}.

We have obtained a tree-partition of width $O(b^2 m^4)=O(m^{24})$. This concludes the proof of Theorem~\ref{thm:tpw-obstructions}.
\end{proof}

\section{Weighted Tree-Partition-Width}
\label{section:weightedtpw}

Recent investigations in algorithmic applications
of tree-partition-width \cite{BodlaenderCW22,BodlaenderMOPvL23} give
algorithms that are fixed parameter tractable with the weighted tree-partition-width as parameter. In these cases, all vertices have weight one, and all
edges have a positive weight. In this section,
we show XALP-completeness for \textsc{Weighted Tree-Partition-Width} 
when all vertices and edges have weight one, and give
an approximation algorithm similar to Theorem~\ref{thm:main}. The latter result can be regarded as a
corollary of Theorem~\ref{thm:main}.

\begin{corollary}
    There is an algorithm that given an $n$-vertex
    graph $G$ with vertex and edge weights, and an
    integer $k$, constructs a tree-partition of
    breadth $O(k^{15})$ for $G$ or reports 
    that $G$ has weighted tree-partition-width more than $k$, in time $k^{O(1)}n^2$.
\end{corollary}

\begin{proof}
    First, observe that when an edge $vw$ has weight more than $k$, then $v$ and $w$ must belong to the same bag in any tree-partition of breadth at most $k$.
    Repeat the following step: if there is
    an edge $vw\in E$ of weight larger than $k$,
    contract the edge. Suppose $x$ is the new
    vertex. Set the weight of $x$ to the sum of
    the weights of $v$ and $w$:
    $w_V(x) = w_V(v)+w_V(w)$.
    If a vertex $y$ is a neighbor to $v$ or $w$,
    then take an edge $xy$, with weight equal
    to the weight of an original edge $vy$ or
    $wy$ if one of these exists, or to the sum
    of these weights if both exist.

    If we have a vertex of weight larger than $k$, we can safely conclude that the weighted
    tree-partition-width of $G$ is larger than $k$,
    and we stop.

    Let $G'$ be the resulting graph. Note that
    all vertex and edge weights in $G'$ are at most $k$. 
    Now, run the third algorithm of Theorem~\ref{thm:main} on $G'$ ignoring all
    edge weights. Note that we can safely have the weight of vertices in the $b$-reduction of $G'$ to be the sum of the weights of the corresponding vertices of $G'$. If this algorithm concludes that
    the tree-partition-width of $G'$ is larger
    than $k$, then we can also conclude that
    the weighted tree-partition-width of $G$ is
    larger than $k$. Otherwise, we obtain
    a tree-partition $(T, (B_i)_{i\in V(T)})$
    of $G'$ of width at most $k$.
    Note that each bag $B_i$ has total weight at most $O(k^7)$. More precisely, in step 4 of the algorithm, we obtain a tree-partition of the $b$-reduction of $G'$ of width $O(wbk^2)=O(k^6)$, with each vertex of weight at most $k$.

    Now, obtain a tree-partition of $G$ by
    undoing all contractions. I.e., if we
    contracted $vw$ to $x$ and $x\in B_i$, then
    replace $x$ by $v$ and $w$ in $B_i$. Repeat this step in the reverse order of which we
    did the contraction steps.
The result is a tree-partition of $G$. 
Undoing contractions does not change the 
total weight of vertices in bags, so the
total weight of vertices in a bag is bounded
by $O(k^7)$. Now, for each pair $ii'$ of adjacent
bags, there are $O(k^{14})$ pairs of vertices
with one endpoint in each bag, and each edge between these bags has weight at most $k$,
so the total weight of edges with one endpoint
in $B_i$ and one endpoint in $B_{i'}$ is bounded
by $O(k^{15})$. Thus, the breadth of the
obtained tree-partition is $O(k^{15})$.
\end{proof}

\begin{corollary}
    There is an algorithm running in polynomial time that
    constructs a tree-partition of breadth
    $O(k^{11}\log^2 k)$ or reports that the
    weighted tree-partition-width is more than $k$.
    
        There is an algorithm running in time $2^{O(k \log k)}n$, that computes a tree-partition of breadth $O(k^{11})$ or reports that the weighted tree-partition-width is more than $k$.
\end{corollary}

\begin{proof}
    Use the same approach as in the previous
    algorithm, but use different subroutines
    to compute the approximate tree-partition, as in the different cases in Theorem~\ref{thm:main}.
\end{proof}

As tree-partition-width is the special case 
of weighted tree-partition-width with all vertex weights one, and all edge weights zero,
we directly have that deciding weighted tree-partition-width is XALP-hard. However, in applications, the case where all edges have
positive weights is of main interest~\cite{BodlaenderCW22,BodlaenderMOPvL23}. In the following, we show that similar hardness also holds in some simple cases with edges having positive weights.
We give a few intermediate results.

\begin{lemma}
    \textsc{Weighted Tree-Partition-Width} is in XALP. \label{lemma:weightedtwpinxalp}
\end{lemma}

\begin{proof}
    This can be proved in the same way as Lemma~\ref{lemma:tpwinxalp}.
\end{proof}

\begin{lemma}
    \textsc{Weighted Tree-Partition-Width} with
    all edges weight 1 is XALP-hard.
\end{lemma}

\begin{proof}
    Take an instance of \textsc{Tree-Partition-Width}: $(G,k)$. Now, set the weight of
    all edges to 1, and all vertices to $L=k^2+1$. Now, for each tree-partition of $G$, its width (where we ignore weights) is at most $k$, if and only if the maximum
    over all bags of the total weight of vertices in a bag is at most $kL$, if and only if its breadth is at most $kL$.
    The latter follows as there are at most $k^2$ edges between bags each of weight one.
    The result now directly follows.
\end{proof}

\begin{theorem}
    \textsc{Weighted Tree-Partition-Width} with
    all vertex and edge weights equal to 1 is XALP-complete.
\end{theorem}

\begin{proof}
    For membership in XALP, see Lemma~\ref{lemma:weightedtwpinxalp}.

    For hardness, we take an instance $(G,k)$ of
    \textsc{Weighted Tree-Partition-Width} with
    all edge weights equal to 1. We may assume that
    all vertices have weight at most $k$, otherwise we have a trivial no-instance.
    
    Now, for every
    vertex $v$ with $w_v=r$, we replace $v$ by a
    subgraph with $r+k+2$ vertices. 
    Take a clique 
    $C_{v,1}$ with $r$ vertices, a second
    clique $C_{v,2}$ with $k$ vertices, and 
    two vertices $x_{v,1}$, $x_{v,2}$.
    
Add the following edges: $x_{v,1}$ and $x_{v,2}$ are adjacent to all vertices in $C_{v,2}$.
Take a vertex $x_{v,3}\in C_{v,2}$ and make
it adjacent to all vertices in $C_{v,1}$.
For each edge $\{v,w\}\in E(G)$, add an edge
between an arbitrary vertex from $C_{v,1}$ and
an arbitrary vertex from $C_{w,1}$. See Figure \ref{fig:enter-label} for an illustration.

Let $G'$ be the resulting graph. All vertices and edges in $G'$ have weight 1.

\begin{figure}
    \centering
    \includegraphics{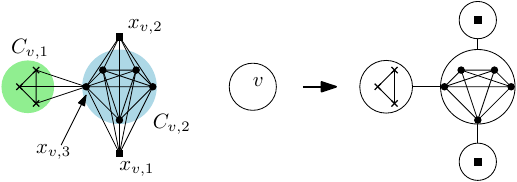}
    \caption{An example of the subgraph for a
    vertex $v$ of weight 3; here $k=5$, and an illustration of how a 
    tree-partition of $G$ is transformed to a 
    tree-partition of $G'$.}
    \label{fig:enter-label}
\end{figure}

The following observation follows directly from the definition.

\begin{claim}
    Let $C$ be a clique in a graph $G$. In each tree-partition of $G$, all vertices in $G$ are in one bag, or in two adjacent bags.
    \label{claim:cliquebags}
\end{claim}

\begin{claim}
    For each vertex $v\in V$ and each tree-partition of $G'$ of breadth at most $k$,
    all vertices of $C_{v,2}$ belong to the same bag.
    \label{claim:samebag2}
\end{claim}

\begin{proof}
    Suppose $(T, (B_i)_{i\in V(T)})$ is a tree
    partition of $G'$ of breadth at most $k$.

    Suppose the vertices in $C_{v,2}$ do not belong to the same bag. Then, by Claim~\ref{claim:cliquebags}, they belong to
    two adjacent bags, say $i_1$ and $i_2$. Now,
    $x_{v,1}$ and $x_{v,2}$ must be in bags, adjacent to $i_1$ and adjacent to $i_2$,
    so $x_{v,1}, x_{v,2} \in B_{i_1} \cup B_{i_2}$. 
    $B_{i_1}\cup B_{i_2}$ contain all $k+2$
    vertices in $C_{v,2}\cup \{x_{v,1},x_{v,2}\}$; thus, within this set, there are
    at least $2k$ pairs of vertices with
    one endpoint in $B_{i_1}$ and one endpoint in
    $B_{i_2}$; however, only one pair of vertices in this set is not adjacent (namely,
    $x_{v,1}$ and $x_{v,2}$), so at least $2k-1$ edges cross the cut from $i_1$ to $i_2$, which contradicts the breadth of the tree-partition when $k>1$.
\end{proof}

\begin{claim}
    For each vertex $v\in V$ and each tree-partition of $G'$ of breadth at most $k$,
    all vertices of $C_{v,1}$ belong to the same bag.
    \label{claim:samebag1}
\end{claim}

\begin{proof}
    Suppose $x_{v,3}\in B_i$. Now, as $x_{v,3}\cup C_{v,1}$ is a clique, all
    vertices in $x_{v,3}\cup C_{v,1}$ are in two incident bags, say $B_i\cup B_{i'}$. As
    $B_i$ contains the $k$ vertices from
    $C_{v,2}$ (by Claim~\ref{claim:samebag2}),
    $C_{v,1}\subseteq B_{i'}$.
\end{proof}

\begin{claim}
    $G$ has weighted tree-partition-width at most $k$, if and only if $G'$ has weighted tree-partition-width at most $k$.
\end{claim}

\begin{proof}
Suppose the weighted tree-partition-width of $G'$ is at most $k$.
    Suppose $(T, (B_i)_{i\in V(T)})$ is
    a tree-partition of breadth at most $k$ of $G'$.
Let for all $i\in V(T)$, $B'_i = \{v\in V ~|~ X_{v,1}\subseteq B_i\}$. It is easy to check that
$(T, (B'_i)_{i\in V(T)})$ is a tree-partition of $G$ of breadth at most $k$. By Claim~\ref{claim:samebag1}, each vertex $v\in V(G)$ belongs to one set $B'_i$. As we replace
a clique with $w_V(v)$ vertices of weight one by one vertex with weight $w_V(v)$, the total weight of each bag is still bounded by $k$.
For each edge $\{v,w\}\in E(G)$, the bags
containing $C_{v,1}$ and $C_{w,1}$ must be the
same or adjacent, as there is an edge between
a vertex in $C_{v,1}$ and a vertex in $C_{w,1}$.
So, the weighted tree-partition-width of $G$ is
at most $k$.

Now, suppose the weighted tree-partition-width of $G$ is
at most $k$.  Suppose $(T, (B_i)_{i\in V(T)})$ is
    a tree-partition of breadth at most $k$ of $G$. Build a tree-partition of $G'$ as
    follows: set $B''_i = \bigcup_{v\in B_i} C_{v,1}$, i.e., we replace each vertex $v$ by the set of vertices $C_{v,1}$, for all $i\in V(T)$.
For each $v\in V(G)$, we add three extra bags:
one bag $i_{v,2}$ containing all vertices
in $C_{v,2}$, one bag containing only the
vertex $x_{v,1}$ which is adjacent to bag $i_{v,2}$, and one bag containing only the
vertex $x_{v,1}$ which also is adjacent to bag $i_{v,2}$. One easily checks that this is a tree-partition of $G'$ of breadth  at most $k$.
\end{proof}

One can check that the transformation can be done with $O(k^{O(1)}+\log n)$ memory. So, we have a pl-reduction from an XALP-complete
problem, and the result follows.
\end{proof}

\section{Conclusion}
\label{section:conclusion}

We settle the question of the exact computation of tree-partition-width, and show that its approximation is tractable. However, many questions remain regarding approximation algorithms:
\begin{itemize}
    \item Is a constant factor approximation tractable?
    \item Can we improve the approximation ratios with similar running times?
\end{itemize}

Some building blocks of the algorithm could possibly be improved in terms of running time.
Is there some value of $b$ polynomial in $w$ and $k$ such that we can compute $G^b$ in time $k^{O(1)}n \log n$ or in time $2^{o(k \log k)}n$ ? This would directly give faster running times for the approximation algorithm, possibly at the cost of a worse approximation ratio.

We gave an algorithm to approximate the weighted
tree-partition-width, which was introduced in
\cite{BodlaenderCW22}, but with a relatively large
factor; we leave it as an open problem to find
better approximations for weighted tree-partition-width.

Another interesting direction is to study the complexity of computing (approximate) tree decompositions on graphs of bounded tree-partition-width.

\bibliographystyle{abbrv}
\bibliography{main}
\label{sec:biblio}

\end{document}